\newcommand{\longversion}[1]{#1}
\newcommand{\shortversion}[1]{}

	\def\year{2020}\relax
	\documentclass[letterpaper]{article} 
	\usepackage{aaai20}  
	\usepackage{times}  
	\usepackage{helvet} 
	\usepackage{courier}  
	\usepackage[hyphens]{url}  
	\usepackage{graphicx} 
	\usepackage{graphicx}  
	\author{ }
	\date{May 2019}

\usepackage{amssymb,amsmath,amsthm,bbm}
\usepackage[usenames,dvipsnames]{color}
\usepackage{makecell}
\usepackage{multirow}

\usepackage[usenames,svgnames,dvipsnames]{xcolor}

\longversion{
\usepackage{fullpage}
\usepackage{charter}
\usepackage{eulervm}
\usepackage{graphicx}
\usepackage[colorlinks=true,linkcolor=blue,citecolor=Mahogany]{hyperref}
}

\allowdisplaybreaks

\usepackage{tikz}
\usetikzlibrary{arrows,positioning}

\newdimen\prevdp
\def\leftlabel#1{\noalign{\prevdp=\prevdepth
   \kern-\prevdp\nointerlineskip\vbox to0pt{\vss\hbox{\ensuremath{#1}}}\kern\prevdp}}

\usepackage{url}

\usepackage{enumerate}

\usepackage{xspace}
\newcommand{\NP}{\ensuremath{\mathsf{NP}}\xspace}
\newcommand{\NPC}{\ensuremath{\mathsf{NP}}-complete\xspace}
\newcommand{\el}{\ensuremath{\ell}\xspace}
\newcommand{\suc}{\ensuremath{\succ}\xspace}

\newcommand{\XTC}{{\sc X3C}\xspace}
\newcommand{\Pb}{\ensuremath{\mathsf{P}}\xspace}

\newcommand{\YES}{{\sc yes}\xspace}
\newcommand{\NO}{{\sc no}\xspace}
\newcommand{\true}{{\sc true}\xspace}

\newcommand{\RGMB}{{\sc Reverse Gerrymandering Bribery}\xspace}
\newcommand{\GMB}{{\sc Gerrymandering Bribery}\xspace}
\newcommand{\MRGM}{{\sc Minimum Reverse Gerrymandering}\xspace}
\newcommand{\MGM}{{\sc Minimum Gerrymandering}\xspace}

\newcommand{\DCP}{{\sc \ensuremath{2}-Disjoint Connected Partitioning}\xspace}

\renewcommand{\AA}{\ensuremath{\mathcal A}\xspace}
\newcommand{\BB}{\ensuremath{\mathcal B}\xspace}

\newcommand{\DD}{\ensuremath{\mathcal D}\xspace}
\newcommand{\EE}{\ensuremath{\mathcal E}\xspace}

\newcommand{\GG}{\ensuremath{\mathcal G}\xspace}
\newcommand{\HH}{\ensuremath{\mathcal H}\xspace}
\newcommand{\II}{\ensuremath{\mathcal I}\xspace}

\newcommand{\LL}{\ensuremath{\mathcal L}\xspace}

\newcommand{\OO}{\ensuremath{\mathcal O}\xspace}
\newcommand{\PP}{\ensuremath{\mathcal P}\xspace}
\newcommand{\QQ}{\ensuremath{\mathcal Q}\xspace}
\newcommand{\RR}{\ensuremath{\mathcal R}\xspace}
\renewcommand{\SS}{\ensuremath{\mathcal S}\xspace}
\newcommand{\TT}{\ensuremath{\mathcal T}\xspace}
\newcommand{\UU}{\ensuremath{\mathcal U}\xspace}
\newcommand{\VV}{\ensuremath{\mathcal V}\xspace}
\newcommand{\WW}{\ensuremath{\mathcal W}\xspace}
\newcommand{\XX}{\ensuremath{\mathcal X}\xspace}

\newcommand{\ZZ}{\ensuremath{\mathcal Z}\xspace}

\newcommand{\NB}{\ensuremath{\mathbb N}\xspace}
\newcommand{\ZB}{\ensuremath{\mathbb Z}\xspace}
\newcommand{\RB}{\ensuremath{\mathbb R^+}\xspace}

\usepackage{nicefrac}

\newtheorem{proposition}{\bf Proposition}

\newtheorem{theorem}{\bf Theorem}

\newtheorem{corollary}{\bf Corollary}
\newtheorem{definition}{\bf Definition}
\newtheorem{probdefinition}{\bf Problem Definition}

\newcommand{\eps}{\ensuremath{\varepsilon}\xspace}
\renewcommand{\epsilon}{\eps}

\newcommand{\ignore}[1]{}

\newcommand{\pr}{\ensuremath{\prime}}

\renewcommand{\ge}{\geqslant}
\renewcommand{\le}{\leqslant}

\usepackage{cleveref}

\crefname{theorem}{Theorem}{Theorems}
\crefname{observation}{Observation}{Observations}
\crefname{lemma}{Lemma}{Lemmata}
\crefname{corollary}{Corollary}{Corollaries}
\crefname{proposition}{Proposition}{Propositions}
\crefname{definition}{Definition}{Definitions}
\crefname{probdefinition}{Problem Definition}{Problem Definitions}
\crefname{claim}{Claim}{Claims}
\crefname{reductionrule}{Reduction rule}{Reduction rules}

\title{Gerrymandering: A Briber's Perspective}

\author{Palash Dey\\Email: palash.dey@cse.iitkgp.ac.in\\Indian Institute of Technology, Kharagpur}

\begin{document}

\maketitle

\begin{abstract}
	We initiate the study of bribery problem in the context of gerrymandering and reverse gerrymandering. In our most general problem, the input is a set of voters having votes over a set of alternatives, a graph on the voters, a partition of voters into connected districts, cost of every voter for changing her district, a budget for the briber, and a favorite alternative of the briber. The briber needs to compute if the given partition can be modified so that (i) the favorite alternative of the briber wins the resulting election, (ii) the modification is budget feasible, and (iii) every new district is connected. We study four natural variants of the above problem -- the graph on voter being arbitrary vs complete graph (corresponds to removing connectedness requirement for districts) and the cost of bribing every voter being uniform vs non-uniform. We show that all the four problems are \NPC even under quite restrictive scenarios. Hence our results show that district based elections are quite resistant under this new kind of electoral attack. We complement our hardness results with polynomial time algorithms in some other cases.
\end{abstract}

\section{Introduction}

A fundamental problem in multiagent systems is to aggregate preferences of a set of agents into a societal preference. Voting has served as one of the most important tool for this aggregation task in various applications~(see for example \cite{PennockHG00,DBLP:conf/fat/ChakrabortyPGGL19}). We assume that agents or voters express their preferences as a complete ranking over some set of alternatives. The plurality voting protocol is arguably the simplest and most widely used voting system where the winners are the set of alternatives who is the most preferred alternative of a maximum number of voters. In this paper we focus on district based election system. In such system, the voters are partitioned into districts. The winner of the election is the alternative which wins in the maximum number of districts. Indeed many real world election systems follow this model: the electoral college in US presidential elections, Indian political election, etc. are important examples of use of district based elections in practice.

However, a typical voting system can come under various kind of election control attacks --- a set of agents, either internal (e.g. voters) or external (e.g. briber), may be able to successfully swing the outcome of the election in their favor. We refer to \cite{DBLP:reference/choice/FaliszewskiR16} for an overview of common control attacks on voting systems. Bartholdi et al.~\cite{bartholdi1989computational} initiated the study of computational complexity of various election control problems and since then it has been one of the major research focus in computational social choice (see \cite[and references therein]{DBLP:reference/choice/ConitzerW16} for example).  Bartholdi et al. and Hemaspaandra et al.~\cite{DBLP:journals/ai/HemaspaandraHR07} studied computational complexity of an important control problem namely {\em "Control by Partitioning Voters into Two Districts."} This fundamental problem has recently been generalized along with two dimensions --- (i) the number of districts can be any integer $k$ which is given as input, (ii) there is a graph on the set of voters and every district required to be a connected subgraph of this graph. This problem is called {\em gerrymandering}~\cite{DBLP:conf/atal/LewenbergLR17,DBLP:conf/atal/Cohen-ZemachLR18}. Indeed there have been serious allegations that some political parties in US
effectively manipulated some elections in their favor through gerrymandering~\cite{erikson1972malapportionment,issacharoff2002gerrymandering}.

\tikzset{
	>=stealth',
	punkt/.style={
		rectangle,
		rounded corners,
		draw=black, very thick,
		text width=22em,
		minimum height=2em,
		text centered},
	pil/.style={
		->,
		thick,
		shorten <=2pt,
		shorten >=2pt,}
}

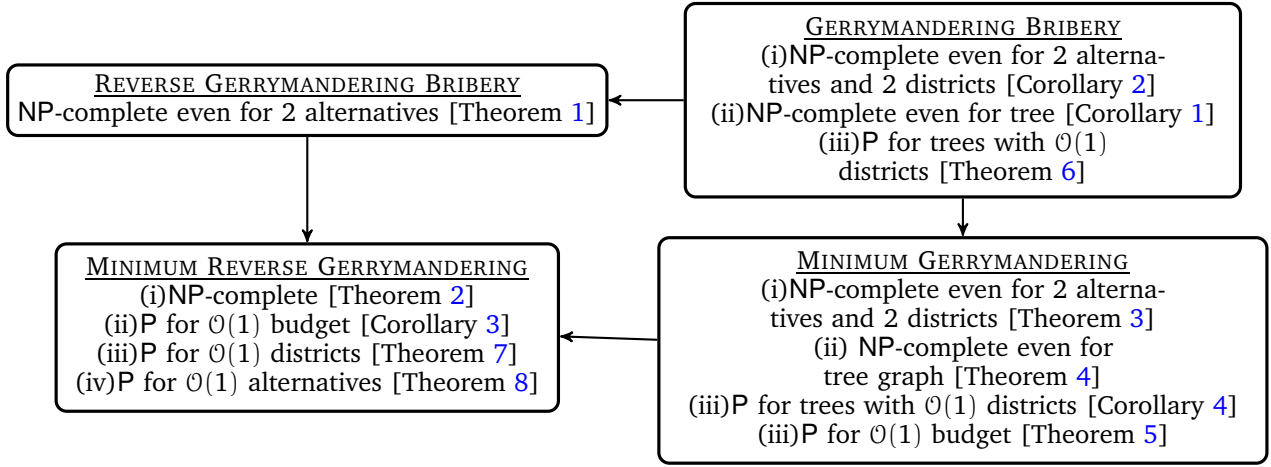
\begin{figure*}[!htbp]
	\centering
	\begin{tikzpicture}[node distance=1cm, auto]
	
	\node[punkt] (rgmb) {\underline{\RGMB}\\\NPC even for $2$ alternatives~[\Cref{thm:rgmb}]};
	
	\node[punkt, inner sep=5pt,below=1.4cm of rgmb,text width=18em]
	(mrgm) {\underline{\MRGM}\\(i)\NPC~[\Cref{thm:mrgm}]\\(ii)\Pb for $\OO(1)$ budget~[\Cref{cor:mrgm_algo_b_cons}]\\(iii)\Pb for $\OO(1)$ districts~[\Cref{thm:mrgm_algo_d_cons}]\\(iv)\Pb for $\OO(1)$ alternatives~[\Cref{thm:mrgm_alt_2}]};
	
	\node[punkt, inner sep=5pt,right=1cm of rgmb, text width=20em]
	(gmb) {\underline{\GMB}\\(i)\NPC even for $2$ alternatives and $2$ districts~[\Cref{cor:gmb}]\\(ii)\NPC even for tree~[\Cref{cor:gmb_2}]\\(iii)\Pb for trees with $\OO(1)$ districts~[\Cref{thm:gmb_algo_d_cons}]};
	
	\node[punkt, inner sep=5pt,below=.5cm of gmb]
	(mgm) {\underline{\MGM}\\(i)\NPC even for $2$ alternatives  and $2$ districts~[\Cref{thm:mgm_2}]\\(ii) \NPC even for tree graph~[\Cref{thm:mgm_tree}]\\(iii)\Pb for trees with $\OO(1)$ districts~[\Cref{cor:mgm_algo_d_cons}]\\(iii)\Pb for $\OO(1)$ budget~[\Cref{thm:mgm_algo_b_cons}]};
	
	\draw[->,thick] (gmb) edge (rgmb);
	\draw[->,thick] (rgmb) edge (mrgm);
	\draw[->,thick] (mgm) edge (mrgm);
	\draw[->,thick] (gmb) edge (mgm);
	\end{tikzpicture}
	\caption{Summary of results and complexity theoretic relationship among problems studied in paper. For two problems $X$ and $Y$, we right $X\rightarrow Y$ to denote that problem $Y$ many-to-one reduces to problem $X$ in polynomial time.}\label{fig:summary}
\end{figure*}

Lev and Lewenberg observed that, although districts being connected is a fundamental requirement for various district based election scenarios like political election, in some other applications, we may get rid of connectedness requirement~\cite{DBLP:conf/aaai/LevL19}. Examples of such applications include election within an organization, election performed over an online platform, etc. Lev and Lewenberg called this problem {\em reverse gerrymandering}.

\subsection{Motivation}

All the existing work on election control by voter partition and gerrymandering study the problem of designing a partition from scratch --- the input is a set of voters and one needs to find a partition favoring some alternative. However, in typical applications of this type of election control, district based political election for example, there already exists a partition of voters into districts and it may not be feasible for someone, let us call her a {\em briber}, to change the existing partition too much. In particular, even if there exists a partition \PP of the voters into districts where a favorite alternative of the briber wins the election, constructing that partition \PP from the existing partition \QQ may require changing the district of too many voters which makes \PP infeasible. Also the effort/cost required for moving a voter from one district to another may depend on the voter and the pair of districts involved. For some voters, it may be infeasible to change her current district. We incorporate these requirements into four computational problems and provide an extensive complexity landscape of these problems.

\subsection{Contribution}

In our most general problem, called \GMB, we are given a partition \PP of the voters, a graph \GG on the voters, a cost function $\pi$ which specifies the cost of moving any voter from a district to another, a favorite alternative $c$ of the briber, and a budget \BB for the briber. The briber needs to compute if there exists another partition \QQ of voters into connected districts which is budget feasible and makes $c$ a plurality winner in the maximum number of districts. We show that the \GMB problem is \NPC even if we have only $2$ alternatives and the graph on the voters is bipartite~[\Cref{cor:gmb}]. However, if the graph on the voters happens to be a tree and the number of districts is only some constant, then we show that there exists a polynomial time algorithm for the \GMB problem~[\Cref{thm:gmb_algo_d_cons}].

Motivated by the concept of reverse gerrymandering, we define and study the \RGMB problem which is the same as the \GMB problem except there is no graph on the voters and consequently the is no requirement for districts to be connected. It seems that existence of a graph on the voters may not be the main reason for \GMB to be intractable since we show that the \RGMB problem too is \NPC even if we have only $2$ alternatives.

We study both the \GMB and \RGMB problem under the assumption that the cost of every transfer is the same ($1$ without loss of generality). We call these problems \MGM and \MRGM respectively. These two problems also capture the robustness of a partition. A partition is be called {\em robust} if many voters need to change their current district to change the winner of the election. We show that the \MRGM problems is \NPC in general~[\Cref{thm:mrgm}] but polynomial time solvable if either briber's budget is a constant~[\Cref{cor:mrgm_algo_b_cons}] or the number of districts is a constant~[\Cref{thm:mrgm_algo_d_cons}] or we have a constant number of alternatives~[\Cref{thm:mrgm_alt_2}]. On the other hand, the \MGM problem turns out to be much harder: it is \NPC even if we have only $2$ alternatives and $2$ districts~[\Cref{thm:mgm_2}]. We also show \MGM is \NPC even if the graph is a tree~[\Cref{thm:mgm_tree}]. We summarize our contribution in this paper in \Cref{fig:summary}.

\subsection{Related Work}

Bartholdi et al.~\cite{bartholdi1989computational} are the first to study, among other election control, the computational problem of making a favorite candidate win by partitioning the voters into two districts. Subsequently, Hemaspaandra et al.~\cite{DBLP:journals/ai/HemaspaandraHR07} studied extensively both the constructive and destructive version of this problem under two tie breaking rule --- tie promoting (TP) and tie eliminating (TE). Lewenberg et al. introduced the gerrymandering problem and showed that gerrymandering is \NPC for election systems where voters in each district first elect a representative and the elected representatives ultimately choose an winner ~\cite{DBLP:conf/atal/LewenbergLR17}. Cohen-Zemach et al. showed that gerrymandering is \NPC for district based election system~\cite{DBLP:conf/atal/Cohen-ZemachLR18}. Ito et al. extensively studied algorithmic aspects of gerrymandering for different graph classes~\cite{DBLP:conf/atal/ItoK0O19}. Lev and Lewenberg studied iterated dynamics which reach to stable equilibrium in the context of reverse gerrymandering where voters change their districts driven by their self interest. Bribery is another well studied election control problem where an external agent, called briber, pays the voters to change/misreport their preference so that a preferred alternative of the briber wins the election. Depending on pricing model of the voters, various notions of bribery have been studied. Prominent models among these variants of bribery include \$bribery, shift bribery~\cite{faliszewski2006complexity,faliszewski2009hard,faliszewski2008nonuniform}, swap bribery~\cite{elkind2009swap}, etc. We refer to~\cite{DBLP:reference/choice/FaliszewskiR16} for an excellent overview of various bribery problems studied in computational social choice.

\section{Preliminaries and Problem Definitions}\label{sec:prelim}

\subsection{Voting Setting}

Let us denote the set $\{1, 2, \ldots, n\}$ by $[n]$ for any positive integer $n$. Let $\AA = \{a_1, a_2, \ldots, a_m\}$ be a set of alternatives and $\VV = \{v_1, v_2, \ldots, v_n\}$ a set of voters. If not mentioned otherwise, we denote the set of alternatives by \AA, the set of voters by \VV, the number of alternatives by $m$, and the number of voters by $n$. Every voter $v_i$ has a preference or vote $\suc_i$ which is a complete order over \AA. We denote the set of complete orders over \AA by $\LL(\AA)$. We call a tuple of $n$ preferences $(\suc_1, \suc_2, \cdots, \suc_n)\in\LL(\AA)^n$ an $n$-voter preference profile. For a preference $\suc\in\LL(\AA)$ and an integer k, we denote the profile consisting of k copies of \suc by $\suc^k$. A map $r:\cup_{n,|\mathcal{C}|\in\mathbb{N}^+}\mathcal{L(C)}^n\longrightarrow 2^\mathcal{C}\setminus\{\emptyset\}$
is called a \emph{voting rule}. One of the simplest and widely used voting rule is the plurality voting rule where the set of winners are the set of those alternatives who appear at the first position of the maximum number of voters' preferences. In the context of plurality voting rule, we say a voter votes for an alternative $x$ if the alternative $x$ is the most preferred alternative of that voter.

In this paper, we consider district based plurality elections only. In this setting, the set \VV of voters are partitioned into $k$ districts as $(P_i)_{i\in[k]}$. Let $W_i\subseteq\AA$ be the set of winners from the $i$-th district $P_i$ according to plurality voting rule. The set \WW of winners of the election is the set of those alternatives which is a winner from the maximum number of districts. An alternative $x\in\AA$ is said to be a {\em unique} winner of the election if $\WW=\{x\}$.

\subsection{Problem Definition}

We now define our problems formally.

\begin{probdefinition}[\GMB]\label{def:gmb}
	Given a set \AA of $m$ alternatives, a set \VV of $n$ voters, a graph \GG on \VV, a preference profile $\PP$ of \VV partitioned into $k$ districts as $(P_i)_{i\in[k]}$, cost functions $\pi:\VV\times[k]\longrightarrow\RB$ specifying cost of moving individual voters to various districts, the briber's budget \BB, and an alternative $c$, compute if it is possible for the briber to construct another partition $\QQ=(Q_i)_{i\in[k]}$ by spending at most \BB so that (i) $\GG[Q_i]$ is connected and (ii) $c$ is the unique plurality winner in the election $\cup_{i\in[k]}\WW_i$ where $\WW_i$ is the set of plurality winners in the $i$-th district $Q_i$. We denote an arbitrary instance of \GMB by $\left(\AA, \VV, \GG, (P_i)_{i\in[k]}, \pi,\BB,c\right)$.
\end{probdefinition}

In \Cref{def:gmb}, the briber wishes that her favorite alternative $c$ wins uniquely in a majority of the districts. Alternatively, the briber could as well wish that $c$ at least co-wins in a majority of the districts. It turns out that all our results (both hardness and algorithmic) extend easily to the co-winner setting. For ease of exposition, we consider only the unique winner setting (as defined in our problem definitions) here. We next define the \MGM problem which is the \GMB problem with the restriction that the cost of every transfer is the same (that is, the function $\pi$ is a constant function $1$).

\begin{probdefinition}[\MGM]\label{def:mgm}
	Given a set \AA of $m$ alternatives, a set \VV of $n$ voters, a graph \GG on \VV, a preference profile $\PP$ of \VV partitioned into $k$ districts as $(P_i)_{i\in[k]}$, the briber's budget \BB, and an alternative $c$, compute if it is possible for the briber to construct another partition $\QQ=(Q_i)_{i\in[k]}$ by moving at most \BB voters so that (i) $\GG[Q_i]$ is connected and (ii) $c$ is the unique plurality winner in the election $\cup_{i\in[k]}\WW_i$ where $\WW_i$ is the set of plurality winners in the $i$-th district $Q_i$. We denote an arbitrary instance of \GMB by $\left(\AA, \VV, \GG, (P_i)_{i\in[k]}, \pi,\BB,c\right)$.
\end{probdefinition}

We next define the \RGMB problem which is the same as the \GMB problem except there is no underlying graph over the voters and thus the districts need not to be connected.

\begin{probdefinition}[\RGMB]\label{def:gsb}
 Given a set \AA of $m$ alternatives, a preference profile $\PP$ of a set \VV of $n$ voters partitioned into $k$ districts as $(P_i)_{i\in[k]}$, cost functions $\pi:\VV\times[k]\longrightarrow\RB$ specifying cost of moving individual voters to various districts, the briber's budget \BB, and an alternative $c$, compute if it is possible for the briber to construct another partition $\QQ=(Q_i)_{i\in[k]}$ by spending at most \BB so that $c$ is the unique plurality winner in the election $\cup_{i\in[k]}\WW_i$ where $\WW_i$ is the set of plurality winners in the $i$-th district $Q_i$. We denote an arbitrary instance of \RGMB by $\left(\AA, (P_i)_{i\in[k]}, \pi,\BB,c\right)$.
\end{probdefinition}

We next define the \MRGM problem which is the \RGMB problem with the restriction that the cost of every transfer is the same (that is, the function $\pi$ is a constant function $1$). 

\begin{probdefinition}[\MRGM]\label{def:mrgm}
	Given a set \AA of $m$ alternatives, a preference profile $\PP$ of a set \VV of $n$ voters partitioned into $k$ districts as $(P_i)_{i\in[k]}$, the briber's budget \BB, and an alternative $c$, compute if it is possible to construct another partition $\QQ=(Q_i)_{i\in[k]}$ by moving at most \BB voters so that $c$ is the unique plurality winner in the election $\cup_{i\in[k]}\WW_i$ where $\WW_i$ is the set of plurality winners in the $i$-th district. We denote an arbitrary instance of \MRGM by $\left(\AA, (P_i)_{i\in[k]}, \BB,c\right)$.
\end{probdefinition}

In \Cref{prop:relation}, we establish complexity theoretic connections among the above problems. In the interest of space, we omit proof of some of our results (they are marked as $*$).

\begin{proposition}\label{prop:relation}\shortversion{[$\star$]}
	We have the following complexity theoretic relationship among our problems: (i) \RGMB reduces to \GMB, (ii) \MRGM reduces to \RGMB, (iii) \MRGM reduces to \MGM, (iv) \MGM reduces to \GMB.
\end{proposition}

\longversion{
\begin{proof}
	\begin{enumerate}[(i)]
		\item Given an instance of \RGMB, define the corresponding instance of \GMB with the same set of alternatives, partition, cost function, briber's budget, and favorite alternative as the \RGMB instance and the graph on voters is defined to be a complete graph. The equivalence of two instances is immediate.
		
		\item Given an instance of \MRGM, define the corresponding instance of \RGMB with the same set of alternatives, partition, briber's budget, and favorite alternative as the \RGMB instance and the cost any transfer is defined to be $1$. The equivalence of two instances is immediate.
		
		\item Given an instance of \MRGM, define the corresponding instance of \MGM with the same set of alternatives, partition, cost function, briber's budget, and favorite alternative as the \MRGM instance and the graph on voters is defined to be a complete graph. The equivalence of two instances is immediate.
		
		\item Given an instance of \MGM, define the corresponding instance of \GMB with the same set of alternatives, partition, briber's budget, and favorite alternative as the \GMB instance and the cost any transfer is defined to be $1$. The equivalence of two instances is immediate.
	\end{enumerate}
\end{proof}
}

\section{Results: Algorithmic Hardness}

We present our intractability results in this section. We begin with showing that the \RGMB problem is \NPC even if we have only $2$ alternatives. For that, we reduce from the \XTC problem which is well known to be \NPC (see~\cite{DBLP:books/fm/GareyJ79} for example). The \XTC problem is defined as follows.
\begin{probdefinition}[\XTC]
	Given a universe $\UU=\{u_i: i\in[3n]\}$ and a collection $\SS=\{S_j:j\in[m]\}$ of subsets of \UU each of cardinality $3$, compute if there exists a sub-collection $\TT\subseteq\SS$ of \SS such that (i) $\cup_{S\in\TT}S = \UU$ and (ii) $|\TT|=n$. We denote an arbitrary instance of \XTC by $\left(\UU=\left\{u_i:i\in[3n]\right\}, \SS=\left\{S_j: j\in[m]\right\}\right)$.
\end{probdefinition}

\begin{theorem}\label{thm:rgmb}
	The \RGMB problem is \NPC even if we have only $2$ alternatives.
\end{theorem}

\begin{proof} 
	The \RGMB problem clearly belongs to \NP. To prove its \NP-hardness, we reduce from the \XTC problem. Let $\left(\UU=\left\{u_i:i\in[3n]\right\}, \SS=\left\{S_j: j\in[m]\right\}\right)$ be an arbitrary instance of \XTC. Without loss of generality, we assume that $5n>m+1$; if not, then we keep adding $3$ new elements in \UU and a set consisting of these $3$ new elements in \SS until $5n$ becomes more than $m+1$. We consider the following instance $\left(\AA, \PP, \pi,\BB,c\right)$ of \RGMB.
	\begin{align*}
		\AA &= \{c,y\}\\
		\PP &= \left((\VV_u)_{u\in\UU},(\VV_S)_{S\in\SS},(D_i)_{i\in[5n-m-1]}\right)\\
		\forall &u\in\UU,\VV_u : \left\{ 3 \text{ votes: } c\suc y\right.\\
			  &  \left.1 \text{ vote: } y\suc c\right\}\\
		\forall &S\in\SS, \VV_S : \left\{ 1 \text{ vote: } c\suc y\right.\\
		&  \left.3 \text{ votes: } y\suc c\right\}\\
		\forall &i\in[5n-m-1], D_i : \left\{ 1 \text{ vote: } c\suc y\right.\\
		&  \left.2 \text{ votes: } y\suc c\right\}
	\end{align*}
	We define the cost function $\pi$. For $u\in\UU$, let $v_u\in\VV_u$ be a voter in the district who votes for the alternative $c$. We define the cost of moving the voter $v_u$ to the district $\VV_S, S\in\SS,$ as $1$ if $u\in S$. The cost of any other movement of voters is $\infty$. Finally we define $\BB=3n$. This finishes the description of the reduced instance. We claim that the two instances are equivalent.
	
	In one direction, let us assume that the \XTC instance is a \YES instance. Without loss of generality, we may assume that $\{S_i:i\in[n]\}$ forms an exact set cover of \UU (by renaming). For every $u\in\UU$ and $i\in[n]$, if $u\in S_i$, then we move the voter $v_u$ to the district $\VV_{S_i}$. We see that the alternative $c$ wins uniquely in $m-n$ districts in the set $\{\VV_S: S\in\SS\}$ and every district in the set $\{\VV_u: u\in\UU\}$. Hence the alternative $c$ wins uniquely in $4n$ districts among $8n-1$ districts and thus the \RGMB instance is also a \YES instance.

	In the other direction, let the \RGMB instance is a \YES instance. Let $\QQ=\left((\VV_u^\pr)_{u\in\UU},(\VV_S^\pr)_{S\in\SS},(D_i^\pr)_{i\in[5n-m-1]}\right)$ be a partition of the voters into $8n-1$ districts such that (i) \QQ can be obtained from \PP by spending at most \BB according to the cost function $\pi$, and (ii) the alternative $c$ wins the election uniquely. We first observe that no voter can be moved to or moved from the district $D_i$ for every $i\in[5n-m-1]$ as the cost of any such movement is $\infty$. So, the alternative $y$ wins in the district $D_i^\pr$ for every $i\in[5n-m-1]$. Therefore, for the alternative $c$ to win the election uniquely, $c$ must win uniquely in at least $n$ districts among the districts in $\{\VV_S^\pr: S\in\SS\}$. We observe that, for the alternative $c$ to win uniquely in any district in $\{\VV_S^\pr: S\in\SS\}$, at least $3$ voters whose preference is $c\suc y$ must move to that district; these voters can only come from the districts in $\{\VV_u:u\in\UU\}$ and there are only $3n$ such voters who can be moved. Hence, there can be at most (and thus exactly) $n$ districts among the districts in $\{\VV_S^\pr: S\in\SS\}$ where the alternative $c$ wins uniquely in the partition \QQ. We claim that $\WW=\{S\in\SS:c \text{ wins uniquely in }\VV_S^\pr\}\subseteq\SS$ forms an exact set cover of \UU; we have already argued that $|\WW|=n$. Suppose \WW does not form an exact set cover for \UU, then there exists an element $z\in\UU$ which \WW does not cover. Since no voter from the district $\VV_z$ have moved, we have $|\WW|<n$ which is a contradiction. Hence the \XTC instance is also a \YES instance.
\end{proof}

In the proof of \Cref{thm:rgmb}, the cost of the movements are extremely non-uniform -- the cost of any movement is either $1$ or $\infty$. We next show that the problem remains \NPC even if all the movements have the same cost (which is the \MRGM problem). However, we need to make the number of alternatives unbounded unlike \Cref{thm:rgmb} where the number of alternatives was $2$. We again reduce from the \XTC problem to prove this result.

\begin{theorem}\label{thm:mrgm}
 The \MRGM problem is \NPC.
\end{theorem}

\begin{proof}
 The \MRGM problem clearly belongs to \NP. To prove its \NP-hardness, we reduce from \XTC. Let $(\UU=\{u_i:i\in[3n]\},\SS=\{S_j: j\in[m]\})$ be an arbitrary instance of \XTC. We consider the following instance $\left(\AA, \PP,\BB,c\right)$ of \MRGM. Let $\lambda$ be any positive integer.
 \begin{align*}
  \AA &= \{a_u: u\in\UU\} \cup\{c\}\\
  \PP &= \left((P_S)_{S\in\SS}, P_C, (P_{u,i})_{u\in\UU,i\in[m-1]}\right)\\
  \forall &S\in\SS, P_S : \lambda \text{ copies of } a_u\suc\cdots, \forall u\in S,\\
                    & \lambda \text{ copies of } c\suc\cdots\lambda\\
  P_C &: m-n+2 \text{ copies of } c\suc\dots\\
  \forall &u\in\UU, i\in[m-1], P_{u,i} : m-n+2 \text{ copies of } a_u\suc\cdots\\
  \BB &= m-n
 \end{align*}
 We claim that the \XTC instance is equivalent to the \MRGM instance.
 
 In one direction, let us assume that the \XTC instance is a \YES instance. Let $\WW\subseteq\SS$ forms an exact set cover for \UU. Then we move one voter from the district $P_C$ to the district $P_S$ for every $S\in\SS\setminus\WW$. We see that, in the resulting partition, the alternative $c$ wins in $m+1$ districts where as every other alternative wins in exactly $m$ districts. This makes $c$ the unique winner in the resulting election. Thus the \MRGM instance is a \YES instance.
 
 On the other direction, let the \MRGM instance is a \YES instance. Let $\QQ=\left((P_S^\pr)_{S\in\SS}, P_C^\pr, (P_{u,i}^\pr)_{u\in\UU,i\in[m-1]}\right)$ be a new partition of the voters such that (i) \QQ can be obtained from \PP by moving at most \BB voters, and (ii) the alternative $c$ wins the election uniquely. We observe that, for every district in $P_C\cup\{P_{u,i}: u\in\UU,i\in[m-1]\}$, one needs to remove/add at least $m-n+2$ votes in order to change the current set of winners. Since the budget \BB is $m-n$, it follows that the winners of these districts do not change even after transferring \BB voters. The alternative $c$ can win from at most $m+1$ districts -- the district $P_C$ and the districts in $\{P_S:S\in\SS\}$. Hence for $c$ to win uniquely, every alternative $a_u, u\in\UU$, should win from at most $m$ districts -- the alternative $a_u$ already wins in every district in $\{P_{u,i}:i\in[m-1]\}$. Let us define $\XX=\{S\in\SS: \exists u\in S, a_u\text{ does not win in } P_S^\pr\}$. Since $\BB=m-n$, it follows that $|\XX|\le m-n$. We claim that $\WW=\SS\setminus\XX$ forms an exact set cover for \UU. We have $|\WW|=|\SS|-|\XX|\ge n$. However, if $|\WW|>n$, then there exists an element $u\in\UU$ such that $u$ belongs to at least $2$ sets in \WW and thus the alternative $a_u$ wins in at least $2$ districts in $\{P_S^\pr: S\in\SS\}$ --- this contradicts our assumption that $c$ wins uniquely in \QQ. Hence we have $|\WW|=n$. We claim that \WW forms an exact set cover for \UU. If not then, there exists an element $v\in\UU$ such that $u$ belongs to at least $2$ sets in \WW and thus the alternative $a_v$ wins in at least $2$ districts in $\{P_S^\pr: S\in\SS\}$ --- this contradicts our assumption that $c$ wins uniquely in \QQ. Hence \WW forms an exact set cover for \UU and thus the \XTC instance is a \YES instance.
\end{proof}

Due to \Cref{prop:relation}, we immediately conclude from \Cref{thm:mrgm} that the \MGM problem is also \NPC. However, we will show a much stronger result --- the \MGM problem is \NPC even if we have only $2$ alternatives and $2$ districts. To prove this result, we reduce from the \DCP problem.

\begin{definition}[\DCP]
	Given a connected graph $\GG=(\VV,\EE)$ and two disjoint nonempty sets $\ZZ_1,\ZZ_2\subset\VV$, compute if there exists a partition $(\VV_1,\VV_2)$ of \VV such that (i) $\ZZ_1\subseteq\VV_1, \ZZ_2\subseteq\VV_2,$ and (ii) $\GG[\VV_1]$ and $\GG[\VV_2]$ are both connected. We denote an arbitrary instance of \DCP by $(\GG,\ZZ_1,\ZZ_2)$.
\end{definition}

It is already known that the \DCP problem is \NPC~\cite[Theorem 1]{DBLP:journals/tcs/HofPW09}.

\begin{theorem}\label{thm:mgm_2}
The \MGM problem is \NPC even if we simultaneously have only $2$ alternatives and $2$ districts.
\end{theorem}

\begin{proof}
 The \MGM problem clearly belongs to \NP. To prove its \NP-hardness, we reduce from the \DCP problem. Let $(\GG^\pr=(\UU,\EE^\pr),\ZZ_1,\ZZ_2)$ be an arbitrary instance of \DCP. Let $z_1$ and $z_2$ be any arbitrary (fixed) vertices of $\ZZ_1$ and $\ZZ_2$ respectively. We consider the following instance $\left(\AA, \GG, \PP=\left(\HH_1,\HH_2\right),\BB,c\right)$ of \MGM.
 \begin{align*}
 \AA &= \{c,y\}\\
 \VV &= \{v_z: z\in\ZZ_1\cup\ZZ_2\}\\ 
 &\cup \{v_u, w_u: u\in\VV\setminus(\ZZ_1\cup\ZZ_2)\}\cup\DD\cup\DD^\pr,\\
 &\DD=\left\{d_i: i\in\left[10n+\left|\ZZ_2\right|\right]\right\},\\
 &\DD^\pr=\{d_i^\pr:i\in[10n+|\ZZ_1|+1]\}\\
 \EE &= \{\{v_a,v_b\}: \{a,b\}\in\EE^\pr\}\\
 &\cup \{\{v_u,w_u\}:u\in\VV\setminus(\ZZ_1\cup\ZZ_2)\}\\
 &\cup\{\{d_i,d_j\}:i,j\in\left[\left|\ZZ_2\right|\right],j=i+1\} \cup \{\{z_2,d_1\}\}\\
 &\cup\{\{d_i^\pr,d_j^\pr\}:i,j\in\left[2n+|\ZZ_1|+1\right],j=i+1\}\\ &\cup\{\{z_1,d_1^\pr\}\}\\
 \HH_2 &= \left\{d_i: i\in\left[\left|\ZZ_2\right|\right]\right\}\\
 \HH_1 &= \VV\setminus\HH_2\\
 \text{Vote}&\text{ of } v_u, u\in\ZZ_2: c\suc y\\
 \text{Vote}&\text{ of } v_u, u\in\VV\setminus\ZZ_2: y\suc c\\
 \text{Vote}&\text{ of } w_u, u\in\VV\setminus(\ZZ_1\cup\ZZ_2): c\suc y\\
 \text{Vote}&\text{ of } d_i, i\in[5n]: y\suc c\\
 \text{Vote}&\text{ of } d_i, 5n<i\le 10n+|\ZZ_2|: c\suc y\\
 \text{Vote}&\text{ of } d_i^\pr, i\in[5n]: y\suc c\\
 \text{Vote}&\text{ of } d_i^\pr, 5n<i\le 10n+|\ZZ_1|+1: c\suc y\\
 \BB&= 2n
 \end{align*}
 We claim that the \DCP instance is equivalent to the \MGM instance.
 
 In one direction, let us assume that the \DCP instance is a \YES instance. Let $(\VV_1,\VV_2)$ be a partition of \UU such that $\ZZ_1\subseteq\VV_1, \ZZ_2\subseteq\VV_2, \GG^\pr[\VV_1]$ and $\GG^\pr[\VV_2]$ are both connected. We consider the following new partition of the voters.
 \begin{align*}
  \text{Voters of } \HH_2^\pr &: \{v_u, w_u: u\in\VV_2\setminus\ZZ_2\}\\
  &\cup\{v_u: u\in\ZZ_2\}\cup\DD^\pr\\
  \text{ voters of } \HH_1^\pr &: \text{ others}
 \end{align*}
 Since $\GG^\pr[\VV_1^\pr]$, $\GG[\DD^\pr]$ is connected, and $\{z_1,d_1^\pr\}\in\EE[\GG]$ is connected, it follows that $\GG[\HH_1^\pr]$ is also connected. Similarly, since $\GG^\pr[\VV_2^\pr]$, $\GG[\DD]$ is connected, and $\{z_2,d_1\}\in\EE[\GG]$ is connected, it follows that $\GG[\HH_2^\pr]$ is also connected. Thus the \MGM instance is also a \YES instance.
 
 In the other direction, let us assume that the \MGM instance is a \YES instance. Let $(\HH_1^\pr, \HH_2^\pr)$ be a new partition formed from $(\HH_1,\HH_2)$ by moving at most \BB ($=2n$) voters such that the alternative $c$ is the unique winner in the new election and each new district remains connected. For the alternative $c$ to win the election uniquely, $c$ must at least co-win in both $\HH_1^\pr$ and $\HH_2^\pr$. Since every voter $d_i^\pr, i\in[n]$ vote for $y$, $\BB=2n$, and the alternative $c$ at least co-wins in $\DD^\pr$, it follows that $\DD^\pr\subseteq\HH_1^\pr$. We claim that every voter $v_u, u\in\ZZ_1$ must belong to $\HH_1^\pr$ since otherwise $c$ will not be a co-winner in $\HH_2^\pr$. Also every voter $v_u, u\in\ZZ_2$ must belong to $\HH_2^\pr$ since otherwise $c$ will not be a co-winner in $\HH_2^\pr$. Let us define $\VV_1=\{u\in\GG^\pr:v_u\in\HH_1^\pr\}$ and $\VV_2=\{u\in\GG^\pr:v_u\in\HH_2^\pr\}$. Since $\HH_1^\pr$ and $\HH_2^\pr$ are both connected, it follows that $\VV_1$ and $\VV_2$ are both connected. Since $(\VV_1,\VV_2)$ forms a partition of $\GG^\pr$, $\ZZ_i\subset\VV_1, \ZZ_2\subset\VV_2,$ it follows that the \DCP instance is also a \YES instance.
\end{proof}

\Cref{thm:mgm_2} immediately gives us the following.

\begin{corollary}\label{cor:gmb_2}
	The \GMB problem is \NPC even if we simultaneously have only $2$ alternatives and $2$ districts.
\end{corollary}

The graph on the voters in the proof of \Cref{thm:mgm_2} can be arbitrarily complex. A natural question would be what is the complexity of the \MGM problem for simple graph classes. 
We show next that \MGM problem is \NPC even if the graph on the set of voters is a tree.
\begin{theorem}\label{thm:mgm_tree}\shortversion{[$\star$]}
	The \MGM problem is \NPC even if the underlying graph is a tree.
\end{theorem}

\longversion{
\begin{proof}
	The \MGM problem clearly belongs to \NP. To show \NP-hardness we reduce from \XTC. Let $(\UU=\{u_i: i\in[3n]\}, \SS=\{S_j: j\in[m]\})$ be an arbitrary instance of \XTC. We consider the following instance $(\AA,\GG,\PP,\BB,c)$ of \MGM. For an element $u\in\UU$, let $f_u$ be the number of sets in \SS which contains $u$.
	\begin{align*}
	 \AA &= \{a_u: u\in\UU\}\cup\{c\}\\
	 \PP &= ((X_S)_{S\in\SS},(Y_{u,i})_{u\in\UU, i\in[m-f_u]})\\
	 \forall & S\in\SS, X_S: 10n  \text{ copies each of } c\suc \cdots, a_u\suc\cdots, u\in S\\
	 \forall & u\in\UU, i\in[m-f_u], Y_{u,i}: 10n \text{ copies of } a_u\suc\cdots\\
	 \BB &= 7n
	\end{align*}
	In the graph \GG on the voters, the induced graph on each district $X_S, S\in\SS, Y_{u,i}, u\in\UU, i\in[m-f_u]$ forms a path. Also the induced graph on $\cup_{u\in\UU}\cup_{i=1}^{m-f_u} Y_{u,i}$ forms a path; one end point of this path be a district $Y_{w,1}^1$ for some $w\in\UU$ with the end vertex (voter) be $Y_{w,1}^1$. In every district $X_S, S=\{x,y,z\}\in\SS$, at one end point of the path, call $X_S^1$, the voter $X_S^1$ vote for $c$ and the other $6$ vertices (voters) closest to $X_S^1$ in the district $X_S$ respectively for $a_x, a_x, a_y, a_y, a_z,$ and $a_z$ respectively. We observe that \GG is a tree. We now claim that the two instances are equivalent.
	
	In one direction, let us assume that the \XTC instance is a \YES instance. Let $\WW\subseteq\SS, |\WW|=n$ forms an exact cover for \UU. For every $S\in\WW$, we move $7$ voters from the district $X_S^1$ to the district $Y_{w,1}$; recall that $2$ of these $7$ voters in $X_S^1$ at the boundary with $Y_{w,1}$ vote for $a_x$ for $x\in S$ and one voter from these $7$ voters vote for $c$. We observe that, in the resulting district $X_S, S\in\WW$, the alternative $c$ wins uniquely. Since \WW forms an exact cover for \UU, it follows that every alternative $a_u, u\in\UU$ wins in $m-1$ districts whereas $c$ wins in $m$ districts. Hence $c$ is the unique winner of the resulting district. Since $|\WW|=n$, we have moved $7n (=\BB)$ voters. Hence the \MGM instance is a \YES instance.
	
	In other direction, let us assume that the \MGM instance is a \YES instance. Let the new districts be $\PP^\pr=((X_S^\pr)_{S\in\SS},(Y_{u,i}^\pr)_{u\in\UU, i\in[m-f_u]})$ where the alternative $c$ wins uniquely. Since $\BB=7n$ and we need to move at least $10n$ voters to change the winner of any district in $\{Y_{u,i}:u\in\UU,i\in[m-f_u]\}$, the winner of $Y_{u,i}$ is the same as $Y_{u,i}^\pr$ for every $u\in\UU, i\in[m-f_u]$. Hence $c$ wins in at most $m$ districts, namely the districts in $\{X_S^\pr: S\in\SS\}$. We define $\WW\subseteq\SS$ as follows: $\WW=\{S\in\SS: \text{ the voter } X_S^1 \text{ belong to } Y_{w,1}\pr\}$. We observe that $|\WW|=n$ and \WW forms a set cover for \UU. Suppose not, then there exists an element $u\in\UU$ such that the alternative $a_u$ co-wins in $f_u$ districts among $\{X_S^\pr: S\in\SS\}$. Hence the alternative $a_u$ wins in $m$ districts in $\PP^\pr$. This contradicts our assumption that $c$ wins uniquely in $\PP^\pr$. Hence the \XTC instance is a \YES instance.
\end{proof}
}

\Cref{thm:mgm_tree} immediately gives us the following corollary because of \Cref{prop:relation}.

\begin{corollary}\label{cor:gmb}
	The \GMB problem is \NPC even if the underlying graph is a tree.
\end{corollary}

\section{Results: Algorithms}

We now exhibit few situations where some of our problems admit polynomial time algorithms. We begin with showing that the \MGM problem is polynomial time solvable if the budget \BB is a constant.

\begin{theorem}\label{thm:mgm_algo_b_cons}\shortversion{[$\star$]}
 The \MGM problem is polynomial time solvable if the budget $\BB$ is a constant.
\end{theorem}

\longversion{
\begin{proof}
 We guess the $b$ voters who leave their current district. We also guess the new district of each of the voters that move. For each guess, we check whether (i) the new districts are connected and (ii) the preferred alternative $c$ wins which can clearly be done in polynomial time. Since there are only $\OO((nk)^b)$ such guesses, our algorithm runs in polynomial time when $b$ is some constant.
\end{proof}
}

\Cref{thm:mgm_algo_b_cons} immediately implies polynomial time algorithm for the \MRGM problem under same setting.

\begin{corollary}\label{cor:mrgm_algo_b_cons}
	The \MRGM problem is polynomial time solvable if the budget $\BB$ is a constant.
\end{corollary}

We next show that there is a polynomial time algorithm for the \GMB problem if we the number of districts is some constant and the graph on the voters happen to be a tree. The intuitive reason for having such an algorithm is that, for trees, any solution of \GMB can be expressed as a set of \el edges that cuts the districts.

\begin{theorem}\label{thm:gmb_algo_d_cons}\shortversion{[$\star$]}
 The \GMB problem is polynomial time solvable for trees if the number of districts is a constant.
\end{theorem}

\longversion{
\begin{proof}
	Let $(\AA,\VV,\GG,\PP,\pi,\BB,c)$ be an arbitrary instance of \GMB where the graph \GG is a tree. Let $d$ be the number of districts which is assumed to be some constant. We observe that, since \GG is a tree, there exists exactly $d$ cut edges with respect to the partition \PP --- edges with end points on different districts. Also, the number of edges in \GG is $n-1$ since \GG is a tree. Hence, the number of partition of \VV into $d$ connected districts is only $\OO((n-1)^d)$ which is $n^{\OO(1)}$ since $d$ is a constant. For every partition, we check if (i) $c$ wins the resulting election uniquely, and (ii) the total cost of all the transfers is at most \BB (both can be done in polynomial time). If we can find such a partition where both the above conditions hold, we output \YES, otherwise we output \NO. Clearly, when the algorithm outputs \YES, it discovers a partition which can be obtained by spending at most \BB and makes $c$ win the election uniquely. On the other hand, when the algorithm outputs \NO, we know that there does not exist any partition which satisfies both the above requirement. Hence the algorithm is correct. 
\end{proof}
}

\Cref{thm:gmb_algo_d_cons} immediately implies a polynomial time algorithm for the \MGM problem in the same setting.

\begin{corollary}\label{cor:mgm_algo_d_cons}
The \MGM problem is polynomial time solvable for trees if the number of districts is a constant.
\end{corollary}

We now show that the \RGMB problem is polynomial time solvable if the number of districts is a constant.

\begin{theorem}\label{thm:mrgm_algo_d_cons}
 The \MRGM problem is polynomial time solvable if the number of districts is a constant.
\end{theorem}

\begin{proof}
 Let $(\AA,\PP=(P_i)_{i\in[\el]},\BB,c)$ be any instance of \MRGM. We present an algorithm check if there exists a partition $\QQ=(Q_i)_{i\in[\el]}$ of voters where the alternative $c$ wins uniquely and it can be obtained from \PP by transferring at most \BB voters. Let $d$ be the number of districts which is assumed to be a constant. We guess the plurality score $w_i, i\in[\el],$ of a winner in every district; there are $\OO(n^d)$ such possibilities. We also guess the subset $\II\subseteq[\el]$ of districts where $c$ is a winner; again there are $\OO(2^d)$ such possibilities. Suppose $|\II|=\lambda$ For these guesses, we execute the following steps. Throughout the steps below, we will decrease \BB sometimes; if \BB becomes negative anytime, then we discard this guess of $w_i, i\in[\el],$ and \II.
 \begin{enumerate}
  \item We initialize a set \RR to empty set.
  
  \item For every alternative $x\in\AA\setminus\{c\}$ and district $i\in[\el]$, if the plurality score $p(x;i)$ of $x$ in the district $i$ is strictly more than $w_i$, then move $p(x;i)-w_i$ voters who vote for the alternative $x$ and having smallest costs from the district $i$ to \RR and decrease \BB by $p(x;i)-w_i$. Throughout the algorithm we will decrease \BB sometimes; if \BB becomes negative anytime, then we discard this guess of $w_i, i\in[\el],$ and \II. For every district $i\in\II$, if the plurality score $p(c;i)$ of $c$ in the district $i$ is strictly more than $w_i$, then move $p(c;i)-w_i$ voters who vote for the alternative $c$ and having smallest costs from the district $i$ to \RR and decrease \BB by $p(c;i)-w_i$. For every district $i\in[\el]\setminus\II$, if the plurality score $p(c;i)$ of $c$ in the district $i$ is more than or equal to $w_i$, then move $p(c;i)-w_i+1$ voters who vote for the alternative $c$and having smallest costs from the district $i$ to \RR and decrease \BB by $p(c;i)-w_i+1$. After this step is executed for every alternative $x\in\AA$ and district $i\in[\el]$, the plurality score of any alternative $x\in\AA$ in any district $i\in[\el]$ is at most $w_i$.
  
  \item We run this step as long as we can. For every district $i\in\II$, if the plurality score of the alternative $c$ in the district $i$ is strictly less than $w_i$, then we do the following: 
  \begin{enumerate}
   \item If there exists a voter in \RR who vote for $c$, then we move that voter from \RR to the district $i$.
   
   \item Else if there exists a district $j\in[\el]\setminus\II$ which has a voter voting for $c$, then we move that voter voting for $c$ and having smallest cost among all voters voting for $c$ in the set $[\el]\setminus\II$ of districts from its current district to district $i$ and decrease \BB by $1$.
   
   \item Else we discard our current guess.
  \end{enumerate}

  \item We run this step as long as we can. For every alternative $x\in\RR$ and district $i\in[\el]$, if the plurality score $p(x;i)$ of $x$ in the district $i$ is at most $w_i-2$, then we move one voter voting for $x$ from \RR to the district $i$.
  
  \item If \RR contains any voter voting for $c$, then we discard our current guess (in this case, $\sum_{i\in\II}w_i + \sum_{i\in[\el]\setminus\II}(w_i-1)$ is strictly less than the sum of plurality scores of $c$ in all districts).
  
  \item We run this step until \RR is non-empty: for every voter $v$ in \RR voting for some alternative $x\in\AA$ and district $i\in[\el]$, if the plurality score $p(x;i)$ of $x$ in the district $i$ is less than $w_i$, then we move $v$ to the district $i$. If for some voter $v$, there does not exist any such district $i$, then we discard our current guess.
  
  \item Now \RR is an empty set. For every alternative $x\in\AA\setminus\{c\}$, if $x$ is a co-winner in at least $\lambda$ (recall $\lambda$ is the guessed number of districts where $c$ has highest plurality score) districts, let $v$ be a voter voting for the alternative $x$ and having smallest cost among all the voters voting for $x$ in all the districts where $x$ has the highest plurality score. Let $j\in[\el]$ be a district such that the plurality score of $x$ in district $j$ is at most $w_j-2$. We move the voter $v$ from his current district to district $j$ and decrease \RR by $1$. If no such district $j$ exists, then we discard our current guess.
 \end{enumerate}
 
 If there exists a guess which is not discarded by any of the above steps, then the algorithm outputs \YES. If all possible guesses are discarded by the above steps, then the algorithm outputs \NO. If the algorithm outputs \YES, it indeed constructs another partition of voters from the given partition by moving at most \BB voters where $c$ is the unique winner in the overall election. On the other hand, if the algorithm discards a guess, indeed there does not exist any partition which respects the guess and can be obtained from the given partition by moving at most \BB voters where $c$ wins. Hence, if every possible guess is discarded by some of the above steps, then the instance is indeed a \NO instance. This concludes the correctness of the algorithm. Since all the above steps can be executed in polynomial time, it follows that our algorithm runs is polynomial time if we have a constant number of districts.
\end{proof}

\begin{theorem}\label{thm:mrgm_alt_2}
	There exists a polynomial time algorithm for the \MRGM problem if the number of alternatives is some constant.
\end{theorem}

\begin{proof}
 Let $(\AA,\PP=(\PP_i)_{i\in[\el]},\BB,c)$ be any instance of \MRGM. We present a dynamic programming based algorithm for a generalization of the \MRGM problem (strictly speaking, our new problem is not a generalization of \MRGM but it will be obvious that \MRGM polynomial time Turing reduces to our new problem for $m=\OO(1)$ as we explain later). In our new problem, other than the input of \MRGM, we are also given a ``supply vector" $(\lambda_a)_{a\in\AA}$ and a ``score vector" $(s(a))_{a\in\AA}$ and we call the instance an \YES instance if and only if there is a partition $\PP^\pr=(\PP_i^\pr)_{i\in[\el]}$ such that the total number of votes for an alternative $a$ is the total number of votes for the alternative $a$ plus $\lambda_a$ ($\lambda_a$ could be negative also), the total number of voters left their district from \PP is at most the given budget \BB, and every alternative $a\in\AA$ wins in $s(a)$ districts in $\PP^\pr$. We have a Boolean dynamic programming table \TT indexed by the set of all tuples in $\{(i, (\lambda_a)_{a\in\AA}, b, (s(a))_{a\in\AA}): i\in\{0,\ldots,\el\}, -\BB\le\lambda_a\le\BB\;\forall a\in\AA,b\in\{0,\ldots,\BB\}, 0\le s(a)\le\el\;\forall a\in\AA\}$. We define $\TT(0,(\lambda_a)_{a\in\AA},b,(s(a))_{a\in\AA})$ to be \true if and only if $\lambda_a=0\;\forall a\in\AA, b\ge 0, s(a)=0\;\forall a\in\AA$. We update the entries of our dynamic programming table as follows. For a vector $(\mu_a)_{a\in\AA}\in\ZB^m$, we define $f^+((\mu_a)_{a\in\AA})=\sum_{a\in\AA: \mu_a>0} \mu_a$ and $f^-((\mu_a)_{a\in\AA})=\sum_{a\in\AA: \mu_a-0} -\mu_a$. Let $(\gamma_a^i)_{a\in\AA}$ be the number of votes that alternative $a$ receives in district $\PP_i$. Lastly, for a vector $(\nu_a)_{a\in\AA}\in\NB^m$ and an alternative $a\in\AA$, we define ${\mathbbm 1}_a((\nu_a)_{a\in\AA})$ to be $1$ if and only if $\nu_a = \max_{x\in\AA} \nu_x$. 
 \begin{align*}
 	&\TT(i,(\lambda_a)_{a\in\AA},b,(s(a))_{a\in\AA})
 	= \bigvee_{\substack{(\mu_a)_{a\in\AA}\in\ZB^m,
 			 	\mu_a\le\nu_a\;\forall a\in\AA\\
 			    f^+((\mu_a)_{a\in\AA})\le b,
 			    f^-((\mu_a)_{a\in\AA})\le b}}\\
 	&\TT\left(i-1,\left(\lambda_a+\mu_a\right)_{a\in\AA}, b-f^+\left(\left(\mu_a\right)_{a\in\AA}\right),\right.\\ &\left.\left(s(a)-{\mathbbm 1}_a\left(\left(\gamma_a^i+\mu_a\right)_{a\in\AA}\right)\right)\right)
 \end{align*}
 
 That is, while we update an entry of our dynamic programming table with $i$ districts, we guess over all possible ways district $i$ could change in a solution; it follows from the above update rule that we need to check $\OO(b^m)$ such guesses. The correctness of our algorithm is immediate from our dynamic programming formulation and update rule. We observe that our dynamic programming table has $\OO(\el(2\BB)^m\BB\el^m)$ entries and each entry can be updated in $\OO(b^m)$ time both of which are polynomial in input parameters if $m=\OO(1)$. Hence our algorithm runs in polynomial time if $m=\OO(1)$.
\end{proof}

\section{Conclusion}

We have initiated the study of minimum gerrymandering where the goal is to minimally change a given partition of voters into districts so that a favored alternative wins the election. Our results show that most of these problems are computationally intractable even under quite restrictive settings. These results show that district based election system is quite robust against against computationally bounded manipulators. We finally complement our hardness results with exhibiting polynomial time algorithms for some of our problems in some special cases. Since most of the problems studied in this paper is \NPC, it would be interesting to study approximation and parameterized algorithms for these problems. For example, many of our problems are polynomial time solvable if the budget or the number of districts are some constant. It would be interesting to study parameterized complexity of our problems with respect to these parameters.

\bibliography{references}
\longversion{\bibliographystyle{alpha}}
\shortversion{\bibliographystyle{aaai}}

\end{document}